\documentclass[final,1p,times]{elsarticle}

\def\equaln#1{{\,\stackrel{#1}{=}\,}}

\usepackage{graphicx}
\usepackage{url}
\usepackage{amsmath}
\usepackage[subrefformat=parens]{subcaption}

\usepackage{amssymb}
\usepackage{amsthm}

\theoremstyle{definition}
\newtheorem{theorem}{Theorem}
\newtheorem{definition}[theorem]{Definition}
\newtheorem{example}[theorem]{Example}
\newtheorem{proposition}[theorem]{Proposition}
\newtheorem{corollary}[theorem]{Corollary}
\newtheorem{remark}[theorem]{Remark}
\newtheorem{lemma}[theorem]{Lemma}

\journal{Elsevier}

\begin{document}

\begin{frontmatter}



\title{Unveiling patterns in xorshift128+ pseudorandom number generators}


\author{Hiroshi Haramoto\corref{cor1}\fnref{label1}}
\address[label1]{Center for Data Science and Faculty of Education, Ehime University,
  3 Bunkyocho, Matsuyama, Ehime, Japan}
\ead{haramoto@ehime-u.ac.jp}
\cortext[cor1]{Corresponding Author}

\author{Makoto Matsumoto\fnref{label2}}
\address[label2]{Graduate School of Advanced Science and Engineering, Hiroshima University, 
  1-3-1 Kagamiyama, Higashi-Hiroshima, Hiroshima, Japan}
\ead{m-mat@math.sci.hiroshima-u.ac.jp}

\author{Mutsuo Saito\fnref{label2}}
\ead{sai10@hiroshima-u.ac.jp}

\begin{abstract}
Xorshift128+ is a newly proposed pseudorandom number generator (PRNG), 
which is now the standard PRNG on a number of platforms. 
We demonstrate that three-dimensional plots of
the random points generated by the generator
have visible structures: they 
concentrate on particular planes in the cube.
We provide a mathematical analysis of this phenomenon. 
\end{abstract}




\end{frontmatter}


\section{Introduction}
\label{sec:intro}
Pseudorandom number generators (PRNG) are basic tools
included in many software, and the recently introduced
xorshift128+ generators \cite{VIGNA2017175}
have become one of the most popular PRNGs.
The xorshift128+ generators passed a stringent test suite 
TestU01 \cite{MR2404400}, consume only
128-bit of memory, and are very fast.
As a result, some of the 
xorshift128+ generators are used as standard generators 
in Google V8 JavaScript Engine 
(\url{https://v8.dev/blog/math-random}),
and hence all the browsers based on this engine use xorshift128+.
Google Chrome, Firefox and Safari are a few examples.

The weakness of xorshift128+, however, was first pointed out by Lemire and O'Neill \cite{LEMIRE},
where it was reported that it fails in tests for $\mathbb{F}_2$-linearity
in BigCrush if the order of the bits in the outputs are reversed.
In this paper, we provide plots demonstrating that certain triplets
of generated pseudorandom points exhibit clear linear patterns, as stated below. 
Figure~\ref{fig:1} shows the 
3D-plot of the points generated by xorshift128+: letting $x_1$, $x_2$, $x_3$, $\ldots \in [0, 1)$ be
uniform pseudorandom real numbers generated by xorshift128+
with parameter $(a, b, c)=(23, 17, 26)$ (this is the standard generator 
in JavaScript V8 Engine), then we plot 
$(x_{3m+1}, x_{3m+2}, x_{3m+3}) \in [0,1)^3$ for $m=0, 1, 2, \ldots$.
After plotting $10000$ points, the points appeared quite random. 

\begin{figure}[h]
  \centering
  \includegraphics[scale=0.3]{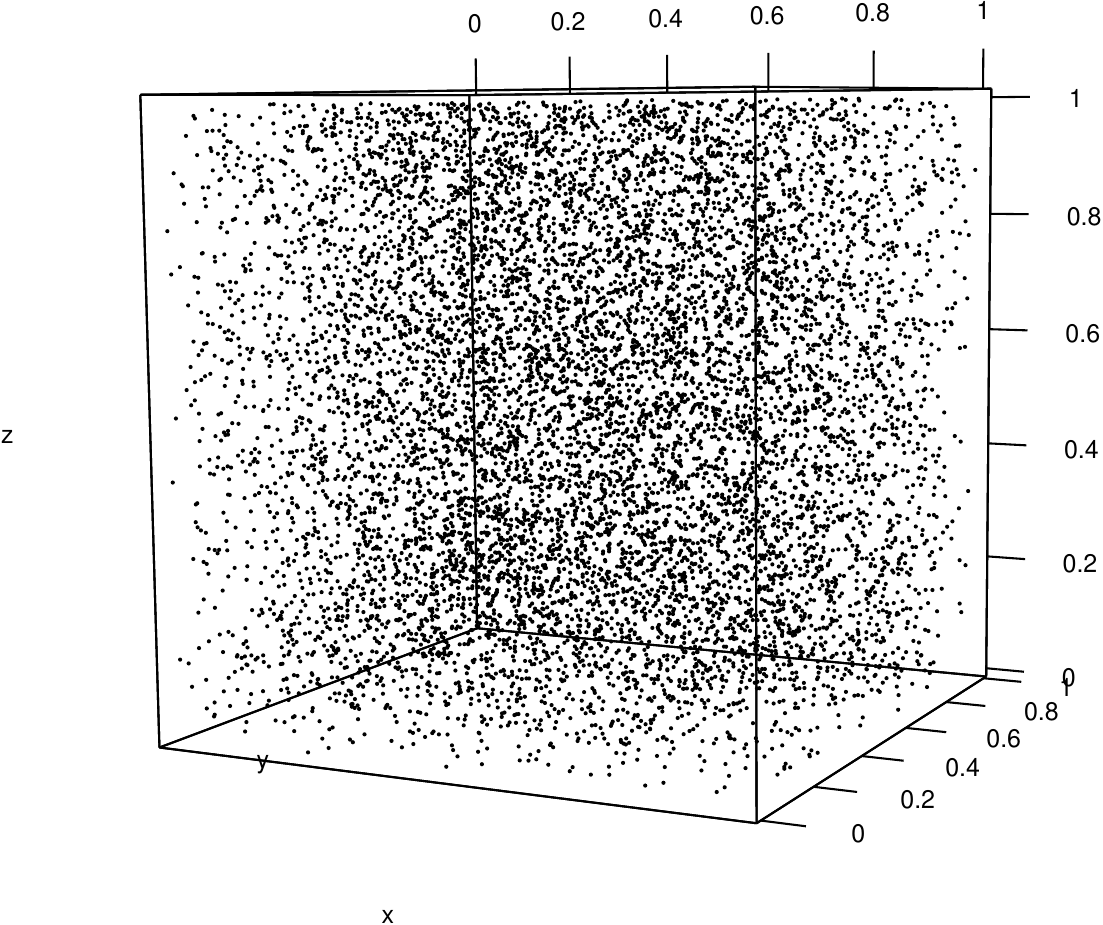}
  \caption{xorshift128+ with $(a,b,c)=(23,17,26)$ 3D random plots for $10000$ points}
  \label{fig:1}
\end{figure}

Thereafter, we magnify the $x$-axis $2^{22}$ times; i.e., we plot 
$(2^{22} x_{3m+1}, x_{3m+2}, x_{3m+3})$
if $2^{22} x_{3m+1} \leq 1$ but skip it if $2^{22} x_{3m+1} > 1$. 
Figure~\ref{fig:2} presents the 3D-plot of the generated $10000$ points 
satisfying the condition $2^{22} x_{3m+1} \leq 1$. 
The points appear non-random; in particular, they concentrate on several planes. 
Such a deviation rarely occurs; it has not been detected
even by BigCrush in TestU01, which is known 
to be one of the most stringent statistical test suites
for PRNGs, because BigCrush does not look at this type of point set. 

\begin{remark}
The plots appear irregular, but one should be careful
about drawing conclusions. Figure~2 shows the plot of points with $x_{3m+1}<1/2^{22}$.
No serious empirical study would select
only a very tiny percentage of variates $x_{3m+1}$
with probability of $<1/2^{22}$ and ignoring the vast 
remaining variates.
\end{remark}
\begin{remark} Xorshift128+ is obtained by breaking the 
$\mathbb{F}_2$-linearity of the xorshift generator \cite{JSSv008i14} by Marsaglia
(see Section~\ref{sec:2}). Since any $\mathbb{F}_2$-linear generator
has a linear pattern among successive variates, the above result
simply says that the proposed method for breaking the $\mathbb{F}_2$-linearity
may not be adequate. It is worth noting that the original xorshift generators fail in
simple statistical tests \cite{xorshift-LEcuyer}.
\end{remark}
We explain what led us to these experiments
in Section~\ref{sec:approximation-by-linear} 
and provide a mathematical explanation for these 
phenomena in Sections~\ref{sec:approximation-by-sum},
\ref{sec:analysis}.
Note that, although we do not go into the details, one can show
that these plane structures can be found everywhere in the
unit cube, not only in the small area where $x<1/2^{22}$
(see Remark~\ref{rem:everywhere}). 

\begin{figure}[htp]
  \centering
  \includegraphics[scale=0.3]{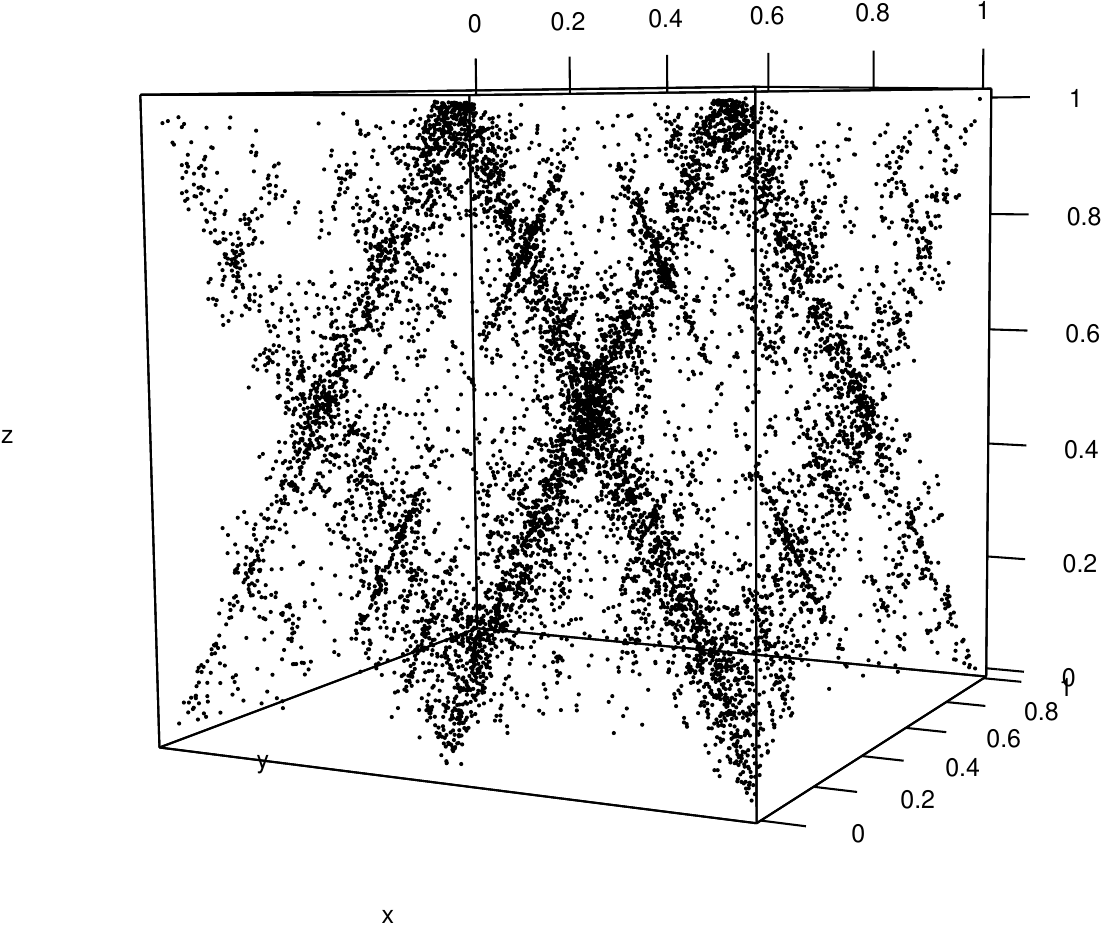}
  \caption{xorshift128+ 3D random plots with $x$-axis magnified by a 
    factor of $2^{22}$, for $10000$ points. The parameter being 
    $(a,b,c)=(23,17,26)$.}
  \label{fig:2}
\end{figure}

Similar pictures for the parameters $(a,b,c)=(26,19,5)$ and $(21,23,28)$ 
are shown in Figure 3. 
There are $8$ parameter sets that have passed TestU01 in \cite{VIGNA2017175}:
$(a,b,c)=
(23, 17, 26)$,
(26, 19, 5),
(23, 18, 5),
(41, 11, 34),
(23, 31, 18),
(21, 23, 28),
(21, 16, 37),
(20, 21, 11),
and each of them exhibits a similar pattern.

\begin{figure}[h]
  \centering
  \includegraphics[scale=0.3]{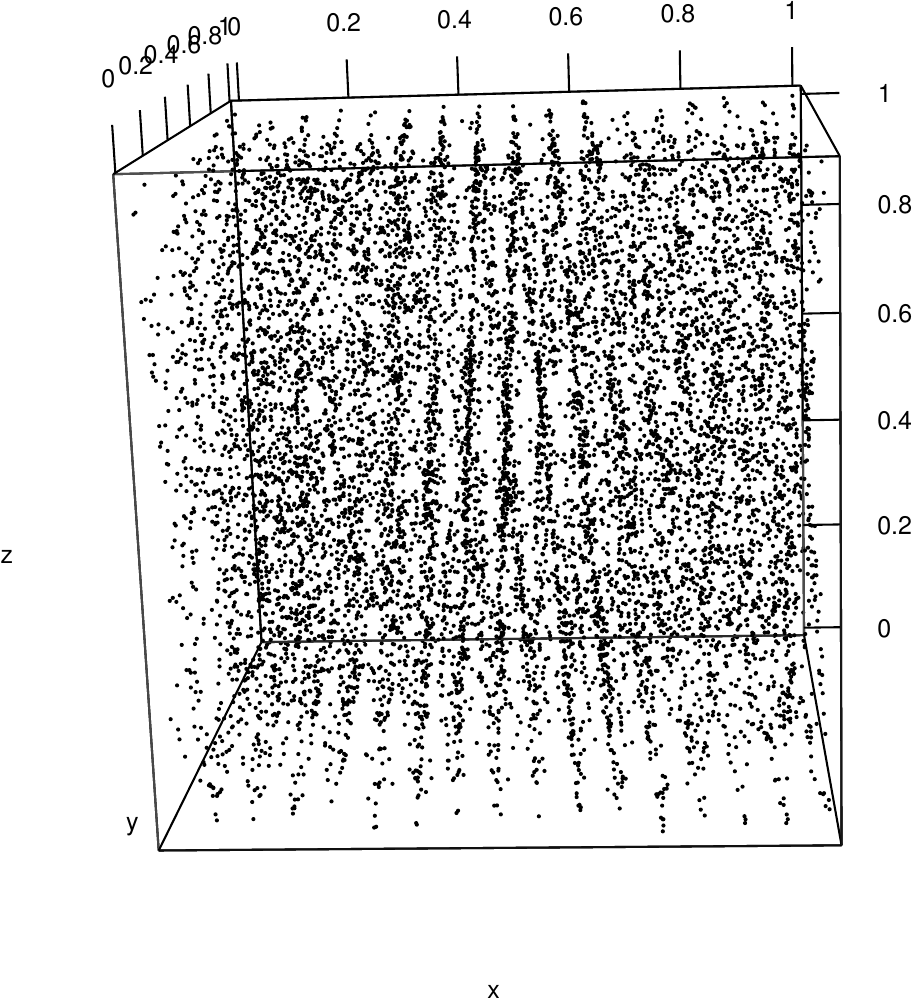}
  \includegraphics[scale=0.3]{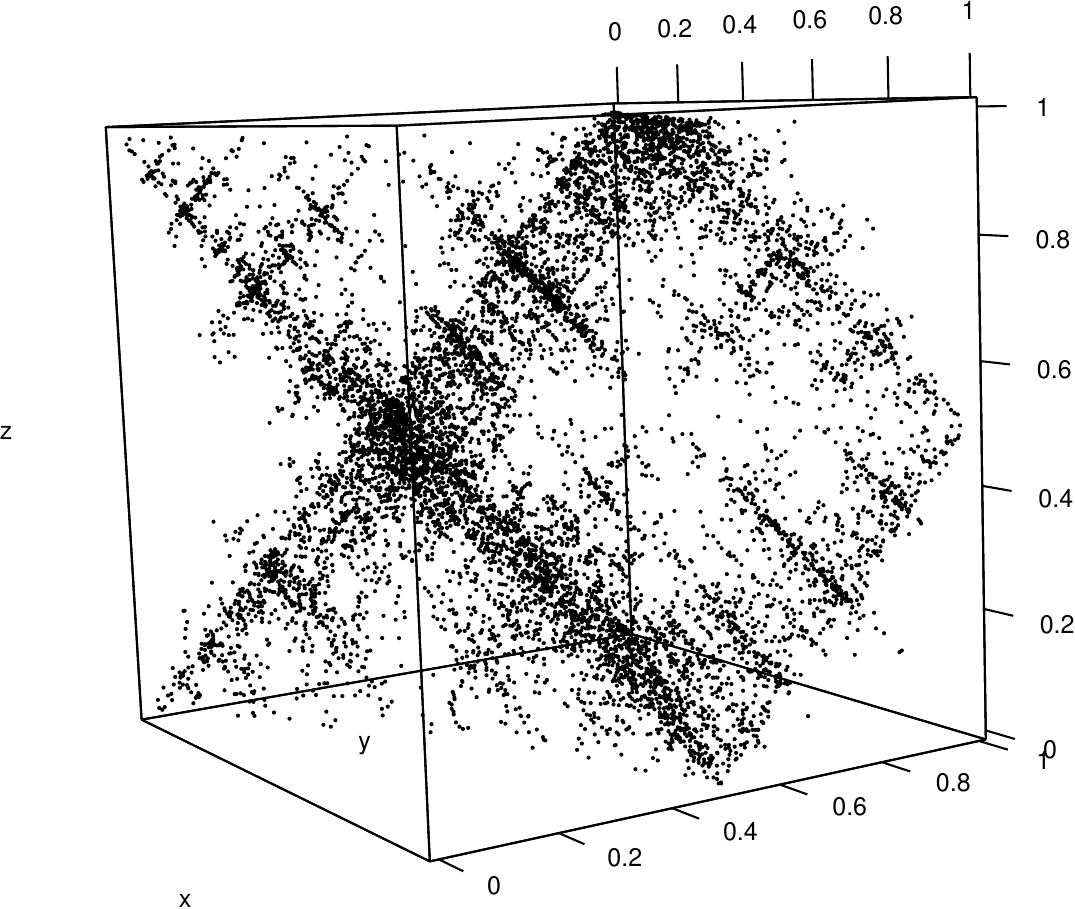}
  \caption{xorshift128+ with $(a,b,c)=(26,19,5)$ (left) and $(21,23,28)$ (right)
    3D random plots of $10000$ points, where $x$-axis is multiplied by $2^{22}$.
  }
  \label{fig:other-parameter}
\end{figure}

\section{Analysis of xorshift128+}\label{sec:2}

\subsection{Description of xorshift128+}
Let $W$ be the set of (unsigned) $\omega$-bit integers.
For most cases in this paper, $\omega$ is 64. 
For $x, y\in W$, $x+y$ means the integer summation 
and for an integer $a$, $ax$ means the integer multiplication;
thus, it may be the case that $x+y \notin W$ and $ax \notin W$.

Let $\mathbb{F}_2 = \{0,1\}$ denote the two-element field and identify 
$W$ with the space of row vectors 
$\mathbb{F}_2^\omega$, where an integer $x \in W$ with
2-adic expansion
$x=x_\omega2^{\omega-1}+\cdots + x_22+x_1$
is identified with $(x_\omega,\ldots, x_1)\in \mathbb{F}_2^\omega$.
For $x,y \in W$, the summation of elements
in $\mathbb{F}_2^\omega$ is denoted by
$x \oplus y \in W$, which is known as the bitwise exclusive-or (xor).
Let $A, B$ be $\omega \times \omega$ matrices with coefficients from $\mathbb{F}_2^\omega$.
Then, $A \oplus B$ denotes the summation of
$\mathbb{F}_2$-matrices, and $xA \in W$ denotes 
the multiplication of 
the $\omega$-dimensional row vector $x \in W$
by the $\omega\times \omega$ matrix $A$ over $\mathbb{F}_2$.

Let $L$ denote the $\omega\times \omega$-matrix
$\begin{pmatrix} 0 & & & & \\ 1 & 0 & & & \\ & 1 & \ddots & & \\
  & & \ddots & \ddots & \\ & & & 1 &0 \end{pmatrix}$.
Then, the multiplication $x \mapsto x L$ is 
$(x_\omega,\ldots,x_2,x_1) \mapsto (x_{\omega-1},\ldots,x_1, 0)$,
which is the so-called left shift; thus, for a non-negative 
integer $a$, $xL^a$ is obtained by $a$-bit left shift
from $x$. 
Similarly, let $R$ denote the matrix 
$\begin{pmatrix} 0 & 1 & & & \\  & 0 & 1 & & \\ & & \ddots & \ddots & \\
  & & & \ddots & 1 \\ & & &  & 0 \end{pmatrix}$, then $x \mapsto xR$
is the right shift. 
Let $I$ denote the identity matrix of size $\omega$. 

Henceforth, we set $\omega=64$.
The xorshift128+ pseudorandom number generator with parameter $(a,b,c)$
is described as follows.
A state consists of two $64$-bit words $(s_i, s_{i+1})$, 
and the next state $(s_{i+1}, s_{i+2})$ is 
determined by the recursion
\begin{equation}
\label{eq:recursion1}
s_{i+2}  = s_i (I \oplus L^a) (I \oplus R^b) \oplus s_{i+1}(I \oplus R^c).
\end{equation}
The initial state $(s_0,s_1)$ is specified by the user, and the 
integer output $o_i$ at the $i$-th state $(s_i, s_{i+1})$ is given by 
\[
o_i = s_i + s_{i+1} \bmod{2^{64}}.
\]
Here, we denote by $x \bmod{2^{64}}$ the 
non-negative integer smaller than $2^{64}$
obtained as the remainder of $x$ divided by $2^{64}$.
To obtain random numbers in $[0,1)$, xorshift128+
outputs $x_i:=o_i/2^{64}$.

\subsection{Rough approximation by an $\mathbb{F}_2$-linear generator}
\label{sec:approximation-by-linear}
We consider a similar generator to xorshift128+, where 
the state transition is identical but for each $i$, the $i$-th output
is replaced with 
\[
 o_i':=s_i \oplus s_{i+1}.  
\]
Since the sequence $s_i, s_{i+1}, \ldots \in W$ is a solution 
of the $\mathbb{F}_2$-linear recursion~(\ref{eq:recursion1}),
the sum $o_i'$ satisfies the same recursion. 
We introduce the following notation: for $x, y\in W$ and 
$0 \leq n \leq \omega$, 
\begin{equation}\label{eq:notation-n}
 x \stackrel{n}{=} y \Leftrightarrow 
(x_\omega, x_{\omega-1},\ldots, x_{\omega-n+1})
=
(y_\omega, y_{\omega-1},\ldots, y_{\omega-n+1}).
 \end{equation}
That is to say, $x\equaln{n} y$ means that the $n$ 
most significant bits (MSBs) of $x$ and $y$
coincide. This is an equivalence relation
satisfying that 
$x\equaln{n} y$
and 
$x'\equaln{n} y'$
implies
$(x\oplus x') \equaln{n} (x\oplus y')$.
Now, put $m:=\min\{b, c\}$. Then, since $o'_i$
satisfies (\ref{eq:recursion1}), it follows that
\begin{equation}\label{eq:aprox}
o'_{i+2} \equaln{m} o'_i (I \oplus L^a) \oplus o'_{i+1},
\end{equation}
because $x R^n\equaln{n} 0$ holds for any $x\in W$.
Our analysis of the behavior of xorshift128+ depends 
on the comparison between $\oplus$, $+$, and $-$.
A detailed discussion is complicated (see Section~\ref{sec:analysis}),
so here we state the rough idea.

\noindent
{\bf Lemma A} Let $x=(x_m,\ldots,x_1), y=(y_m,\ldots,y_1)$ be uniform
random variates over $m$-bit unsigned integers.
Then, $x\oplus y=x+y$ holds if and only if  $x_i+y_i\leq 1$ for every $1\leq i \leq m$.
The equality $x \oplus y=x-y$ ($x \oplus y=-x+y$, respectively)
holds if $x_i\geq y_i$ for $1\leq i \leq m$ (if $x_i\leq y_i$ for $1\leq i \leq m$, respectively).
Thus, if $x, y$ are uniformly random, then $x\oplus y=x+y$ holds 
with probability $(3/4)^m$.

This is a restatement of Lemmas~\ref{theoremplus},~\ref{theoremminus}
in the next section.

To describe the idea for the analysis of xorshift128+, from now on in this subsection, 
we impose a very extreme assumption: 
in the process of computing $o_i, o_{i+1}, o_{i+2}$ from $s_i$ and $s_{i+1}$, every $+$-operation
results in the same value as with the $\oplus$-operation, and vice versa.
From the viewpoint of Lemma A, this happens rarely, but this is how 
the above-mentioned experiments with $x$-axis
magnified can be attained; moreover, generalizing this idea and evaluating errors  
allows us to analyze plane structures, as described in Section~\ref{sec:analysis}.  

On the right hand side of (\ref{eq:aprox}), the term
$o'_i (I \oplus L^a)=o'_i \oplus o'_iL^a = o'_i \oplus (2^a o'_i \bmod{2^{64}})$
is, under the assumption,
$(1+2^a)o'_i \bmod{2^{64}}$. 
Thus, the right hand side of (\ref{eq:aprox}) is 
 $(1+2^a)o'_i+o'_{i+1} \bmod{2^{64}}$. 
Since the $i$-th output $o_i$ of xorshift128+
is by assumption $o'_i$,  it follows that
$o_{i+2}\equaln{m}(1+2^a)o_i+o_{i+1}$ holds.
This means that one can obtain a simple linear relation between 
$((1+2^a)x_{3m+1}, x_{3m+2}, x_{3m+3})$.
This serves us the motivation for the experiment presented in 
Figure~\ref{fig:2}, depicting
the points $2^{-64}(\mu o_{3m+1}, o_{3m+2}, o_{3m+3})$
satisfying $\mu o_{3m+1} \leq 2^{64}$ for 
the parameter set $(a,b,c)=(23,17,26)$, under a magnifying factor of $\mu=2^{a-1}$.

The factor $(1+2^a)$ obtained in the above analysis
is replaced with $\mu=2^{a-1}$
because, as a factor smaller than $(1+2^a)$, it is expected to reveal $(1+2^a)/\mu$ copies of the
planes. In fact, Figure~\ref{fig:2} suggests that
the plane-structures are repeated twice in the direction 
of the $x$-axis. The left picture in Figure~\ref{fig:other-parameter},
where $a=26$ and $\mu=2^{22}$ (thus $\mu$ is approximately $1/16$ of $1+2^a$)
shows the structure is repeated 16 times in the direction 
of the $x$-axis.

There are many problems with these simplified arguments because $(3/4)^m$
decreases with increasing $m$, 
and there are a number of xors involved in the transition function. We treat
these problems in the next subsections.
\subsection{Approximation of exclusive-or by sum and subtraction}
\label{sec:approximation-by-sum}
First we analyze a relation between 
$x\oplus y$ and $\pm x\pm y$,
which is actually not used in the 
analysis of xorshift128+ but 
would be helpful in understanding the
next step of analyzing $x \oplus y \oplus z$, which 
is necessary for the investigating of xorshift128+.
Let $x$, $y$ be $n$-bit unsigned integers. 
Our claim is that if $x$ and $y$ are uniformly random unsigned integers, then $x \oplus y$,  
with a significant probability, coincides with
one of $x+y$, $x-y$, or $y-x$, in the sense that the probability increases
if $x \oplus y$ is replaced by a purely random $n$-bit integer.
The operation $x \oplus y$ is similar to $x+y$, except that
no carry is reflected. 
Consequently $x \oplus y \leq x+y$ holds, and the equality holds 
if and only if no carry occurs. 
Equivalently, if and only if $(x_i, y_i) \in \{(0,0), (1,0), (0,1)\}$ holds
for $i=1$, $2$, $\ldots$, $n$. 
Thus, among $4^n$ possibilities of the pairs $(x, y)$, 
exactly $3^n$ pairs satisfy $x \oplus y = x+y$. 

\begin{lemma}[xor equals sum]
  \label{theoremplus}
  Let $x=\sum\limits_{i=1}^{n} x_i 2^{i-1}$
  and $y=\sum\limits_{i=1}^{n} y_i 2^{i-1}$ be two $n$-bit integers. 
  We identify $x,y$ with vectors 
  $(x_n,\ldots,x_1)$, $(y_n,\ldots,y_1) \in \mathbb{F}_2^n$,
  respectively.
  Then it follows that $x \oplus y \leq x+y$ with equality 
  if and only if $(x_i, y_i) \neq (1,1)$ for $i=1$,$2$,$\ldots$, $n$. 
  Thus, among $4^n$ pairs $(x, y)$, $3^n$ pairs satisfy the equality. 
  Similarly, among $4^n$ pairs $(x, y)$, 
  $4\cdot 3^{n-1}$ pairs satisfy the equality
  $x \oplus y = x+y \bmod{2^n}$. 
\end{lemma}

\begin{proof}
The count for the pairs satisfying $x\oplus y=x+y$
follows from the above observation. For $x\oplus y=x+y \bmod{2^n}$, 
observe that the condition ``$(x_i,y_i)\neq (1,1)$ for
$i=1,\ldots,n-1$'' is necessary and sufficient.
\end{proof}

Let us highlight another observation regarding when the equality 
\[
x \oplus y = x-y
\]
occurs. Again, we do not consider the right hand side modulo $2^n$.
If we compute the subtraction $x-y$ in binary without borrowing, 
we obtain $x \oplus y$. There may be borrows, 
leading to the inequality
\[
x \oplus y \geq x-y, 
\]
with equality when no borrow occurs for each digit, 
or equivalently, the pair of bits $(x_i, y_i)$ lies in 
$\{(0,0), (1,0), (1,1)\}$ for each $i=1,2,\ldots,n$. 
There are $3^n$ such pairs $(x, y)$.

\begin{lemma}[xor equals subtraction]
  \label{theoremminus}
  Let $x, y \in \mathbb{F}_2^n$ be as in Theorem \ref{theoremplus}.
  We have the inequality
  \[
  x \oplus y \geq x-y,
  \]
  with equality if and only if $(x_i, y_i) \neq (0,1)$ 
  for each $i=1,2,\ldots, n$, and there are $3^n$ such pairs.
  The equality
  $x \oplus y = x-y \bmod{2^n}$
  holds if and only if $(x_i, y_i) \neq (0,1)$ 
  for each $i=1,\ldots, n-1$, and there are $4\cdot 3^{n-1}$ such pairs.
\end{lemma}

\begin{lemma}
  \label{counttheorem}
  Let $X$ be the set of pairs 
  $\{(x, y) \mid x, y \in \mathbb{F}_2^n\}$, and put
  \begin{align*}
    A&:=\{(x, y) \in X \mid x\oplus y = x+y\} \\
    B&:=\{(x, y) \in X \mid x\oplus y = x-y\} \\
    C&:=\{(x, y) \in X \mid x\oplus y = y-x\}.
  \end{align*}
  Then, $\# X = 4^n$, $\#A=\#B=\#C=3^n$, 
  $\#(A \cap B) = \#(B \cap C) = \#(C \cap A) = 2^n$, and 
  $\#(A \cap B \cap C)=1$ hold. 
  In particular, $\#(A\cup B \cup C)=3\cdot 3^n - 3 \cdot 2^n + 1$ holds.
  Similarly, if we put
  \begin{align*}
    A'&:=\{(x, y) \in X \mid x\oplus y = x+y \bmod{2^n}\} \\
    B'&:=\{(x, y) \in X \mid x\oplus y = x-y \bmod{2^n}\} \\
    C'&:=\{(x, y) \in X \mid x\oplus y = y-x \bmod{2^n}\}, 
  \end{align*}
then
  $\#A'=\#B'=\#C'=4\cdot 3^{n-1}$, 
  $\#(A' \cap B') = \#(B' \cap C') = \#(C' \cap A') = 4 \cdot 2^{n-1}$, and 
  $\#(A' \cap B' \cap C')=4$ hold. 
  In particular, 
   $\#(A'\cup B' \cup C')=4(3\cdot 3^{n-1} - 3 \cdot 2^{n-1} + 1)$ holds.
\end{lemma}

\begin{proof}
  We have 
  \begin{align*}
    A&=\{(x, y) \in X \mid 
    (x_i, y_i) \in \{(0,0), (0,1), (1,0)\} \mbox{ for } i=1,\ldots,n \} \\
    B&=\{(x, y) \in X \mid 
    (x_i, y_i) \in \{(0,0), (1,0), (1,1)\} \mbox{ for } i=1,\ldots,n \} \\
    C&=\{(x, y) \in X \mid 
    (x_i, y_i) \in \{(0,0), (0,1), (1,1)\} \mbox{ for } i=1,\ldots,n \}, 
  \end{align*}
  and $\#A=\#B=\#C=3^n$ holds. Since
  \[
  A \cap B = \{(x, y) \in X \mid 
  (x_i, y_i) \in \{(0,0), (1,0)\} \mbox{ for } i=1,\ldots,n \},
  \]
  $\#(A \cap B) = 2^n$ follows, and 
  $\#(B \cap C) = \#(C \cap A) = 2^n$ is similarly proved.
  We have 
  $A \cap B \cap C = \{(x, y) \in X \mid 
  (x_i, y_i) \in \{(0,0)\} \mbox{ for } i=1,\ldots,n \}$, and 
  see that $\#(A \cap B \cap C) = 1$. 
  The equality $\#(A\cup B \cup C)=3\cdot 3^n - 3 \cdot 2^n + 1$
  follows from the standard 
  inclusion-exclusion principle.
  The statements for $A',B',C'$ follow from a similar argument, 
  where the conditions on the most significant bits should be neglected.
\end{proof}

\begin{example}
  Suppose $n=3$, i.e., we consider three-bit precision. 
  Then, among $4^3=64$ pairs $(x, y)$, 
  we showed that at least one of $x\oplus y =x+y, x-y, y-x$ 
  occurs for $3\cdot3^3-3\cdot2^3+1=58$ pairs:
  only $64-58=6$ exceptions exist. 
  For $n=4$, among $256$ pairs, $196$ pairs satisfy one of the three relations.
\end{example}

In the analysis of xorshift128+, we do not use the above lemma, which was 
only proved to show that xor is close to sum and subtraction.
For xorshift128+, we use
the following lemma and propositions.

\begin{lemma}\label{th:three-sum}
Let $Y$ be the set of triples of $n$-bit integers 
$\{(u, v, w) \mid u, v, w \in \mathbb{F}_2^n\}$.
We denote $u=(u_n,\ldots,u_1)$, 
$v=(v_n,\ldots,v_1)$, and 
$w=(w_n,\ldots,w_1)$.
Let $p,q,r$ be elements in $\{1, -1\}$.
Then, the equality 
$$
u\oplus v \oplus w = pu+qv+rw
$$
holds if and only if 
$$pu_i+qv_i+rw_i\in \{0,1\}$$
holds for $1\leq i \leq n$.
In particular, 
\begin{description}
 \item[Case~1]
 $u\oplus v \oplus w = u+v+w$
holds if and only if 
the triple $(u_i,v_i,w_i)$ has at most one 1 for each $i$.
There are four such triples 
$(0,0,0), (0,0,1), (0,1,0), (1,0,0)$
among the eight possible triples.
 \item[Case~2-1]
 $u\oplus v \oplus w = -u+v+w$
holds if and only if 
the triple $(u_i,v_i,w_i)$ is one of the six triples
$(0,0,0), (0,0,1), (0,1,0), (1,0,1), (1,1,0), (1,1,1)$
for $1\leq i\leq n$.
This is therefore satisfied by $6^n$ elements among the $8^n$ elements in $Y$.
 \item[Case~2-2] the case 
 $u\oplus v \oplus w = u-v+w$ is similar to the above,
with the first component and the second component exchanged.
 \item[Case~2-3] the case 
 $u\oplus v \oplus w = u+v-w$ is similar to Case~2-1,
with the first component and the third component exchanged.

 \item[Case~3] Case $u\oplus v \oplus w = u-v-w$ occurs
if and only if 
the triple $(u_i,v_i,w_i)$ is one of the four triples
$(0,0,0),  (1,0,0), (1,0,1), (1,1,0)$
for $1\leq i\leq n$. Similar results hold for
$-u+v-w, -u-v+w$. 
\end{description}
The equality 
$$
u\oplus v \oplus w = pu+qv+rw \mod 2^{n}
$$
holds if and only if 
$$pu_i+qv_i+rw_i\in \{0,1\}$$
holds for $1\leq i \leq n-1$.
Thus, if two of $p,q,r$ are $1$ and 
the other one is $-1$, then 
$6^{n-1}$
triples satisfy the equality. 
\end{lemma}
\begin{proof}
Assume that $u\oplus v \oplus w=u+v+w$.
Consider the $i$-th bit $u_i$, $v_i$, $w_i$.
Their summation $\bmod 2$ is equal to $u_i\oplus v_i \oplus w_i$.
The $i$-th bit of $u+v+w$
is $u_i+ v_i + w_i$ plus the carry from the below. 
For $i=1$, the equality $u_i\oplus v_i \oplus w_i=u_i+v_i+w_i \bmod 2$
holds. If this summation produces a carry, then 
the inequality $u_2\oplus v_2 \oplus w_2 \neq u_2+v_2+w_2 +1 \bmod 2$, 
contradicts our assumption and, thus, there is no 
carry at the $1$-st bit; for that reason, 
$u_2\oplus v_2 \oplus w_2 = u_2+v_2+w_2 \bmod 2$ holds, 
and if the right hand side is greater or equal to two, in which case 
the carry nullifies the equality of the next bit
in a similar way. It follows that $u_i+v_i+w_i\in \{0,1\}$
holds for each $i$; conversely, if $u_i+v_i+w_i\in \{0,1\}$ holds
for each $i$, then $u\oplus v \oplus w=u+v+w$.

For the case $u\oplus v \oplus w=u+v-w$,
a similar proof beginning with the least significant bit
leads to the condition $u_i+v_i-w_i\in \{0,1\}$ for each $i$,
since the carry/borrow is at most one, which breaks the
equality modulo 2 of the next bit. Other cases follow similarly.

For the case $u\oplus v \oplus w=pu+qv+rw \bmod 2^{64}$,
the proof follows in a similar manner, except that 
we should neglect the carry or borrow at the most significant bit.
\end{proof}

We now introduce some notations used in the next section.
\begin{definition}\label{def:n}
Let $n$ be an integer $1\leq n \leq 64$.
For a 64-bit integer $u$, let us denote by $[u]_n$ 
the 64-bit integer having the same $n$ MSBs with $u$
and the rest $(64-n)$ bits being zeroes, 
and by $(u)_n$ the $n$-bit integer given by 
the $n$ MSBs of $u$.
Define $\Delta_n(u):=u-[u]_n$. This is 
the integer obtained from the
$64-n$ LSBs of $u$, and hence 
$0\leq \Delta_n(u) \leq 2^{64-n}-1$.
Note that $[u\oplus v]_n=[u]_n\oplus [v]_n$ holds.
\end{definition}

\subsection{Analysis of the three dimensional correlation in xorshift128+}
\label{sec:analysis}
We consider the three consecutive outputs
\begin{align*}
  x &= s_{i}+s_{i+1} \bmod{2^{64}} \\
  y &= s_{i+1}+s_{i+2} \bmod{2^{64}} \\
  z &= s_{i+2}+s_{i+3} \bmod{2^{64}},
\end{align*}
where $s_{i+2}$ and $s_{i+3}$ are determined from 
$(s_i, s_{i+1})$ by the recursion (\ref{eq:recursion1}). 
We show that $(x,y,z)$ concentrates on some planes
as follows. We have 
\begin{align}
  \notag z&=s_{i+2}+s_{i+3} \bmod{2^{64}} \\ 
  \label{z1}&=(s_{i+1}(I\oplus R^c) 
  \oplus s_i(I\oplus L^a)(I\oplus R^b)) \\
  \label{z2}&\quad +(s_{i+2}(I\oplus R^c) 
  \oplus s_{i+1}(I\oplus L^a)(I\oplus R^b)) \bmod{2^{64}}.
\end{align}

The numbers $b$ and $c$ are no less than $5$ in the 8 parameter sets
listed in Section~\ref{sec:intro}. Define $m:=\min\{b,c\}$,
and take an integer $1\leq n \leq m$.
Then
$u(I+R^b) \equaln{n} u \equaln{n} u(I+R^c)$
holds for any 64-bit integer $u$ (see the notation (\ref{eq:notation-n})). 
Thus, we have
\begin{align}
  \label{zfirst}
    s_{i+1}(I\oplus R^c) \oplus s_i(I\oplus L^a)(I\oplus R^b)
    &~ \equaln{n} ~  s_{i+1} \oplus s_i(I\oplus L^a)  \\
  \label{zsecond} 
    s_{i+2}(I\oplus R^c) \oplus s_{i+1}(I\oplus L^a)(I\oplus R^b)
    &~ \equaln{n} ~  s_{i+2} \oplus s_{i+1}(I\oplus L^a).
\end{align} 
Now, set
\begin{align}
\label{eq:err7}
\epsilon_{7,n}
&:=\Delta_n\left(s_{i+1}(I\oplus R^c) \oplus s_i(I\oplus L^a)(I\oplus R^b)\right), \\
\label{eq:err8}
\epsilon_{8,n}
&:=\Delta_n\left(s_{i+2}(I\oplus R^c) \oplus s_{i+1}(I\oplus L^a)(I\oplus R^b)\right),
\end{align}
and thus $0\leq \epsilon_{7,n}, \epsilon_{8,n}\leq 2^{64-n}-1$ holds.
From (\ref{zfirst}), (\ref{zsecond}), 
(\ref{eq:err7}), and (\ref{eq:err8}), we have 
\begin{align}
\notag
s_{i+1}(I\oplus R^c) \oplus s_i(I\oplus L^a)(I\oplus R^b)&=
\left[s_{i+1}(I\oplus R^c) \oplus s_i(I\oplus L^a)(I\oplus R^b)\right]_n
\\
\notag
&\quad +\Delta_n\left(s_{i+1}(I\oplus R^c) \oplus s_i(I\oplus L^a)(I\oplus R^b)\right)
 \\
\label{eq:err7'}
&= [s_{i+1} \oplus s_i(I\oplus L^a)]_n + \epsilon_{7,n}, \mbox{ and } \\
\notag
s_{i+2}(I\oplus R^c) \oplus s_{i+1}(I\oplus L^a)(I\oplus R^b)&=
\left[s_{i+2}(I\oplus R^c) \oplus s_{i+1}(I\oplus L^a)(I\oplus R^b)\right]_n
\\
\notag
&\quad +\Delta_n\left(s_{i+2}(I\oplus R^c) \oplus s_{i+1}(I\oplus L^a)(I\oplus R^b)\right)\\
\label{eq:err8'}
&= [s_{i+2} \oplus s_{i+1}(I\oplus L^a)]_n + \epsilon_{8,n}.
\end{align}
Thus by (\ref{z1}), (\ref{z2}), (\ref{eq:err7'}), and (\ref{eq:err8'}) we have
\begin{align}
\notag z &= s_{i+1}(I\oplus R^c) \oplus s_i(I\oplus L^a)(I\oplus R^b)
 + s_{i+2}(I\oplus R^c) \oplus s_{i+1}(I\oplus L^a)(I\oplus R^b) \bmod 2^{64}\\
\notag   &=
 \left[s_{i+1} \oplus s_i(I\oplus L^a)\right]_n 
 + \left[s_{i+2} \oplus s_{i+1}(I\oplus L^a)\right]_n + \epsilon_{7,n} + \epsilon_{8,n}\bmod 2^{64}\\
\notag
&=
 [s_{i+1} \oplus s_i \oplus s_iL^a]_n 
 + [s_{i+2} \oplus s_{i+1}\oplus s_{i+1}L^a]_n + \epsilon_{7,n} + \epsilon_{8,n} \bmod 2^{64} \\
&=
 ([s_{i+1}]_n \oplus [s_i]_n \oplus [s_iL^a]_n) 
 + ([s_{i+2}]_n \oplus [s_{i+1}]_n \oplus [s_{i+1}L^a]_n) 
 + \epsilon_{7,n} + \epsilon_{8,n} \bmod 2^{64}.
\label{eq:err78}
\end{align}

\begin{proposition}\label{prop:approx-xyz}
Let $(p,q,r)\in \{1,-1\}^3$, and $n\leq m=\min\{b,c\} \leq 64$.
Assume that two equalities
\begin{equation}\label{eq:three-sum-i}
[s_{i+1}]_n\oplus [s_i]_n\oplus [s_iL^a]_n 
=p[s_{i+1}]_n+q[s_i]_n+r[s_iL^a]_n \bmod 2^{64}
\end{equation}
\begin{equation}\label{eq:three-sum-i1}
[s_{i+2}]_n\oplus [s_{i+1}]_n\oplus [s_{i+1}L^a]_n 
=p[s_{i+2}]_n+q[s_{i+1}]_n+r[s_{i+1}L^a]_n \bmod 2^{64}
\end{equation}
hold in (\ref{eq:err78}).
Then
\begin{equation}\label{eq:plane-error}
z = py + qx + r\cdot (2^ax \bmod 2^{64})
-p(\epsilon_{1,n}+\epsilon_{4,n})
-q(\epsilon_{2,n}+\epsilon_{5,n})
-r(\epsilon_{3,n}+\epsilon_{6,n})
+\epsilon_{7,n}+\epsilon_{8,n} \bmod 2^{64}
\end{equation}
holds for $0\leq \epsilon_{i,n} \leq 2^{64-n}-1$ ($i=1,\ldots,8$).
Suppose that two of $p,q,r$ are $+1$
and the other one is $-1$. Then, the error term
is bounded by
$$
|-p(\epsilon_{1,n}+\epsilon_{4,n})
-q(\epsilon_{2,n}+\epsilon_{5,n})
-r(\epsilon_{3,n}+\epsilon_{6,n})
+\epsilon_{7,n}+\epsilon_{8,n}|
\leq 4(2^{64-n}-1).
$$
\end{proposition}
\begin{proof}
Define
$\epsilon_{1,n}:=\Delta_n(s_{i+1})$,
$\epsilon_{2,n}:=\Delta_n(s_i)$,
$\epsilon_{3,n}:=\Delta_n(s_iL^a)$,
$\epsilon_{4,n}:=\Delta_n(s_{i+2})$,
$\epsilon_{5,n}:=\Delta_n(s_{i+1})$, and
$\epsilon_{6,n}:=\Delta_n(s_{i+1}L^a)$.
Then 
$s_{i+1}=[s_{i+1}]_n+\epsilon_{1,n}$,
$s_i=[s_i]_n+\epsilon_{2,n}$,
$s_iL^a=[s_iL^a]_n+\epsilon_{3,n}$,
and so on.
From (\ref{eq:err78}), (\ref{eq:three-sum-i}), and (\ref{eq:three-sum-i1})
we have
\begin{align}
\notag
z &=
p[s_{i+1}]_n+q[s_i]_n+r[s_iL^a]_n +p[s_{i+2}]_n+q[s_{i+1}]_n+r[s_{i+1}L^a]_n
+\epsilon_{7,n}+\epsilon_{8,n} 
\bmod 2^{64} \\
\notag
&=
p(s_{i+1}-\epsilon_{1,n})+q(s_i-\epsilon_{2,n})
+r((s_iL^a)-\epsilon_{3,n}) \\
\notag
&\quad
+p(s_{i+2}-\epsilon_{4,n})+q(s_{i+1}-\epsilon_{5,n})
+r((s_{i+1}L^a)-\epsilon_{6,n})
+\epsilon_{7,n}+\epsilon_{8,n} 
\bmod 2^{64} \\
\notag
&=
p(s_{i+1}+s_{i+2})+q(s_i+s_{i+1})+r(s_iL^a+s_{i+1}L^a) \\
\label{eq:eighterr}
&\quad
-p(\epsilon_{1,n}+\epsilon_{4,n})
-q(\epsilon_{2,n}+\epsilon_{5,n})
-r(\epsilon_{3,n}+\epsilon_{6,n})
+\epsilon_{7,n}+\epsilon_{8,n} \bmod 2^{64}.
\end{align}
Making the substitutions
$x=s_i+s_{i+1} \bmod 2^{64}$, 
$y=s_{i+1}+s_{i+2}\bmod 2^{64}$,
$s_iL^a=2^a s_i \bmod 2^{64}$,
and 
$2^ax \bmod 2^{64}=(s_iL^a+s_{i+1}L^a) \bmod 2^{64}$
gives (\ref{eq:plane-error}).

For the last statement, under the condition that 
two of $p,q,r$ are $+1$ and the other one is $-1$, the eight 
$\epsilon_{i,n}$ have four $+1$'s and four $-1$'s
as coefficients, and thus we obtain the desired bound.
\end{proof}

\begin{corollary}\label{cor:approx}
If the assumptions of Proposition~\ref{prop:approx-xyz}
on $s_i$, $s_{i+1}$, $s_{i+2}$ are satisfied and
two of $p,q,r$ are $+1$ and the other one is $-1$, then
\begin{equation}
 |z - (q + r\cdot 2^a)x - py \mod 2^{64}|\leq 4(2^{64-n}-1)
\end{equation}
holds, and consequently the point $(x,y,z)$
plotted by xorshift128+ tends to be near the plane
$$
z=(q + r\cdot 2^a) x + py \mod 2^{64}.
$$
\end{corollary}
\begin{remark}\label{rem:approx}
The normalized error (where the output range $[0,2^{64})$ is normalized to 
$[0,1)$ by dividing $2^{64}$)
is bounded by $4(2^{64-n}-1)/2^{64} \leq 2^{2-n}$. Thus, for $n\geq 5$,
the concentration on planes would be visible as this normalized error
bound is $1/8$. To obtain a rough estimation of 
the distribution of the error term, after normalization, 
we might regard that $\epsilon_{i,n}$
are independently uniformly distributed among $[0,2^{-n})$
(although they may not be independent in reality).
Since the error term is the summation of eight such random variables,
it is approximated by a normal distribution with mean $0$ 
(owing to the four pluses and four minuses in the coefficients) 
and variance $8 \times (1/12 \cdot 2^{-2n})=2/3 \cdot 2^{-2n}$ 
and standard deviation $\sqrt{2/3}\cdot 2^{-n}$.
Consequently, an approximate estimation of the 
probability that the error is smaller than $\sqrt{2/3}\cdot 2^{-n}$ 
is 68\%, using the theory of normal distribution. Under this assumption,
for example if $n=5$, the pattern would be clearly visible.
\end{remark}

We shall compute the probability that the conditions in 
Proposition~\ref{prop:approx-xyz} are satisfied
when two of $p,q,r$ are $+1$ and the other one is $-1$.
\begin{proposition}\label{prop:10}
Suppose that $n\leq \min\{b,c\}\leq a$ holds
and that the initial seeds $s_0, s_1$
are independently uniformly selected from the set of 
64-bit integers.
Let two of $p,q,r$ be $+1$ and the other one be $-1$.
The probability that both events 
(\ref{eq:three-sum-i}) and (\ref{eq:three-sum-i1})
occur
for $(p,q,r)=(-1,1,1)$ is 
$(5/8)^{n-1}$. 
The probability for $(p,q,r)=(1,-1,1)$ is $(1/2)^{n-1}$.
The same holds for $(p,q,r)=(1,1,-1)$.
Thus, the 3D outputs $(x,y,z)$ of xorshift128+ are 
close to the following planes (in the sense of Corollary~\ref{cor:approx} 
and Remark~\ref{rem:approx}):
\begin{itemize}
 \item 
 $z=(1+2^a)x -y \mod 2^{64}$ with probability $(5/8)^{n-1}$,
 \item 
   $z=(-1+2^a)x +y \mod 2^{64}$ with probability $(1/2)^{n-1}$, and 
 \item 
   $z=(1-2^a)x +y \mod 2^{64}$ with probability $(1/2)^{n-1}$.
\end{itemize}
\end{proposition}

\begin{remark} \label{rem:everywhere} $ $
\begin{enumerate}
\item The condition $5\leq \min\{b,c\}\leq a$
holds for the eight recommended parameter sets
of xorshift128+ in \cite{VIGNA2017175} as stated 
in \S\ref{sec:intro}, 
and consequently we may take $n\leq 5$.
\item For example if $n=5$, 
$(5/8)^{n-1}=0.15258...$ and $(1/2)^{n-1}=0.0625.$
These probabilities are rather high.
\item This proposition implies that such concentration 
to the planes (depicted in Figure~\ref{fig:2} for
only small values of $x$)
occurs everywhere in the cube.
\end{enumerate}
\end{remark}

\begin{proof}[Proof of Proposition~\ref{prop:10}.]
Since the recursion is bijective, for any fixed $i$,
$s_i$ and $s_{i+1}$ are independent and uniform.
Put $u:=(s_i)_n$, $v:=(s_iL^a)_n$, $w:=(s_{i+1})_n$, 
$s:=(s_{i+1}L^a)_n$ (see Definition~\ref{def:n}). 
Since $n\leq \min\{b,c\}\leq a$ holds,
these are independent and uniform. 
Now, put $t:=(s_{i+2})_n$.
By (\ref{eq:recursion1})
and (\ref{zfirst})
$t=(s_{i+2})_n=(s_{i+1})_n\oplus (s_i)_n\oplus (s_iL^a)_n=u\oplus v \oplus w$.
Then, the results follow from the following proposition.
\end{proof}

\begin{proposition}
Let $n$ be a positive integer, and $u,v,w,s$ be $n$-bit
integers. Suppose that $u,v,w,s$ are independently and 
uniformly randomly distributed. 
Let two of $p,q,r$ be $+1$ and the other one be $-1$.
Put
$$
t:=u \oplus v \oplus w.
$$
Consider the following two statements.
\begin{description}
 \item[A1] $u \oplus v \oplus w=pu + qv + rw \mod 2^{n}$.
 \item[A2] $w \oplus s \oplus t=pw + qs + rt \mod 2^{n}$.
\end{description}
Then, the following hold.
\begin{enumerate}
 \item 
 Case $(p,q,r)=(-1,1,1)$. Then both A1 and A2 hold
with probability $(5/8)^{n-1}$.
 \item
 Case $(p,q,r)=(1,-1,1)$ or $(1,1,-1)$. Then, both A1 and A2 hold
with probability $(1/2)^{n-1}$.
\end{enumerate}
\end{proposition}
\begin{proof} $ $

\noindent
Case $(p,q,r)=(-1,1,1)$. By Lemma~\ref{th:three-sum}, A1 holds
if and only if $(u_i,v_i,w_i)$ is one of 
$$(0,0,0), (0,0,1), (0,1,0), (1,0,1), (1,1,0), (1,1,1)$$
for every $1\leq i \leq n-1$. Thus, A1 holds with probability
$(6/8)^{n-1}$. Suppose that A1 holds. 
We discuss the conditional distribution of $t_i$. 
Suppose $w_i=0$. Then, among the above-mentioned six possibilities, the three possible cases
$(0,0,0), (0,1,0), (1,1,0)$ occur with equal probability.
Consequently, 
$t_i=-u_i+v_i+w_i$ is $0$ with probability $2/3$,
and $1$ with probability $1/3$.
Suppose $w_i$=1. Then, the three possible cases
$(0,0,1), (1,0,1), (1,1,1)$ occur with equal probability.
Consequently, $t_i=-u_i+v_i+w_i$ is $0$ with probability $1/3$
and $1$ with probability $2/3$.
Then, each of $(w_i,t_i)=(0,0), (0,1), (1,0), (1,1)$ occurs
with probability $1/2$ times $2/3, 1/3, 1/3, 2/3$, respectively.
Since $s_i$ is independent, 
$(w_i,s_i,t_i)$ is one of
$$(0,0,0), (0,0,1), (0,1,0), (1,0,1), (1,1,0), (1,1,1)$$
with probability $1/4\cdot (2/3+1/3+2/3+2/3+1/3+2/3)=10/12=5/6$.
Thus, under the condition that A1 holds, A2 occurs with 
probability $(5/6)^{n-1}$. Altogether, 
both A1 and A2 occur with probability
$$
(6/8)^{n-1}\cdot (5/6)^{n-1}=(5/8)^{n-1}.
$$

\noindent
Case $(p,q,r)=(1,-1,1)$. A1 holds
if and only if $(u_i,v_i,w_i)$ is one of 
$$(0,0,0), (0,0,1), (1,0,0), (0,1,1), (1,1,0), (1,1,1)$$
for every $1\leq i \leq n-1$. This occurs with probability
$(6/8)^{n-1}$. Suppose that A1 holds. 
Suppose $w_i=0$. Then, the three possible cases
$(0,0,0), (1,0,0), (1,1,0)$ occur with equal probability,
and $t_i=u_i-v_i+w_i$ is $0$ with probability $2/3$,
and $1$ with probability $1/3$.
Suppose $w_i$=1. Then, the three possible cases
$(0,0,1), (0,1,1), (1,1,1)$ occur with equal probability,
and $t_i=u_i-v_i+w_i$ is $0$ with probability $1/3$,
and $1$ with probability $2/3$.
Then, each of $(w_i,t_i)=(0,0), (0,1), (1,0), (1,1)$ occurs
with probability $1/2$ times $2/3, 1/3, 1/3, 2/3$, respectively.
Since $s_i$ is independent, 
$(w_i,s_i,t_i)$ is one of 
$$(0,0,0), (0,0,1), (1,0,0), (0,1,1), (1,1,0), (1,1,1)$$
with probability $1/4\cdot (2/3+1/3+1/3+1/3+1/3+2/3)=8/12=2/3$.
Thus, under the condition that A1 holds, A2 occurs with 
probability $(2/3)^{n-1}$. Thus, 
both A1 and A2 occur with probability
$$
(6/8)^{n-1}\cdot (2/3)^{n-1}=(1/2)^{n-1}.
$$

\noindent
Case $(p,q,r)=(1,1,-1)$. A1 holds
if and only if $(u_i,v_i,w_i)$ is one of 
$$(0,0,0), (0,1,0), (1,0,0), (0,1,1), (1,0,1), (1,1,1)$$
for every $1\leq i \leq n-1$. This occurs with probability
$(6/8)^{n-1}$. Suppose that A1 holds. 
Suppose $w_i=0$. Then, the three possible cases
$(0,0,0), (0,1,0), (1,0,0)$ occur with equal probability,
and $t_i=u_i+v_i-w_i$ is $0$ with probability $1/3$,
and $1$ with probability $2/3$.
Suppose $w_i$=1. Then, the three possible cases
$(0,1,1), (1,0,1), (1,1,1)$ occur with equal probability,
and $t_i=u_i+v_i-w_i$ is $0$ with probability $2/3$,
and $1$ with probability $1/3$.
Then, each of $(w_i,t_i)=(0,0), (0,1), (1,0), (1,1)$ occurs
with probability $1/2$ times $1/3, 2/3, 2/3, 1/3$, respectively.
Since $s_i$ is independent, 
$(w_i,s_i,t_i)$ is one of 
$$(0,0,0), (0,1,0), (1,0,0), (0,1,1), (1,0,1), (1,1,1)$$
with probability $1/4\cdot (1/3+1/3+2/3+2/3+1/3+1/3)=8/12=2/3$.
Thus, under the condition that A1 holds, A2 occurs with 
probability $(5/6)^{n-1}$. Altogether, 
both A1 and A2 occur with probability
$$
(6/8)^{n-1}\cdot (2/3)^{n-1}=(1/2)^{n-1}.
$$
\end{proof}

The other 
cases $(p,q,r)=(1,1,1), (1,-1,-1), (-1,1,-1), (-1,-1,1)$
in Proposition~\ref{prop:approx-xyz} can be analyzed
similarly but require more complicated arguments, 
and the resulting probability that both A1 and A2 hold
(or equivalently (\ref{eq:plane-error}) holds) is
$(1/4)^{n-1}$, $(1/4)^{n-1}$, $(3/8)^{n-1}$,
$(1/8)^{n-1}$, respectively, for these four cases.
In tandem with the analysis provided in the previous section, 
we see that 
the consecutive three outputs $(x,y,z)$ of xorshift128+
tend to lie near eight planes:
\begin{equation}\label{eq:64}
 z= \pm ((1\pm 2^a) x \pm y) \bmod{2^{64}},
\end{equation}
which gives an explanation 
for the visible structure in Figure~\ref{fig:2}.
We compare these planes with the outputs of xorshift128+.

Figure~\ref{fig:4planes} describes four planes 
\begin{equation}\label{eq:normalized}
Z = \pm (1+2^{23}) X \pm Y \bmod{1}
\end{equation}
with restriction $0 \leq X \leq 1/2^{23}$, $0 \leq Y \leq 1$. 
We normalize variables $x, y, z$ into $X, Y, Z$ by 
$X=x/2^{64}$, $Y=y/2^{64}$, $Z=z/2^{64}$.
Then, the ambient cube in the equation (\ref{eq:64}) (whose edge length is $2^{64}$)
is normalized into a unit cube, and hence can be compared with
Figure~\ref{fig:2}.

The $X$-axis is magnified by the factor $2^{23}$.
The other four planes with coefficient $-1+2^{23}$ are 
very close to those for $1+2^{23}$ and, thus, has been omitted.
Each plane consists of two connected components in this region.
Figure~\ref{fig:union} shows the union of these four planes.

\begin{figure}[h]
  \centering
  \begin{minipage}[b]{0.4\linewidth}
    \centering
    \includegraphics[scale=0.4]{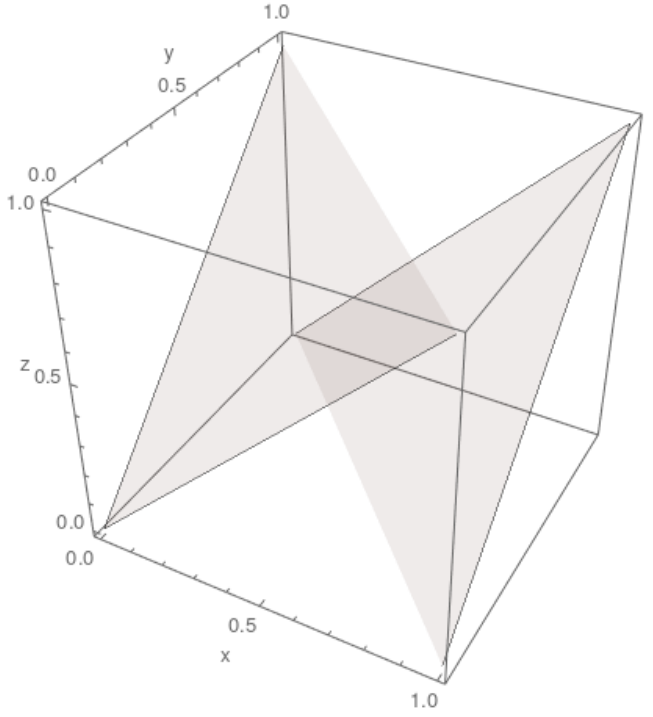}
    \subcaption{}
  \end{minipage}
  \begin{minipage}[b]{0.4\linewidth}
    \centering
    \includegraphics[scale=0.4]{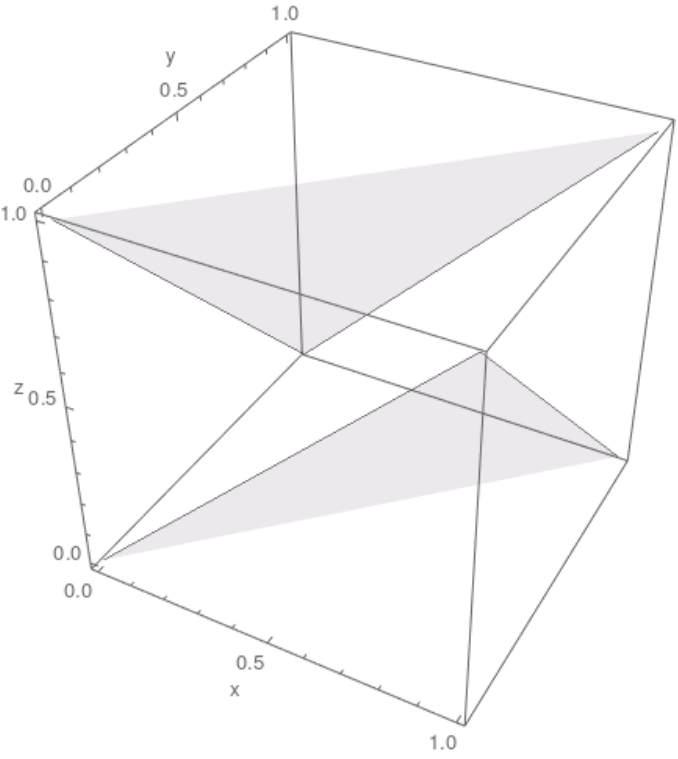}
    \subcaption{}
  \end{minipage}

  \begin{minipage}[b]{0.4\linewidth}
    \centering
    \includegraphics[scale=0.4]{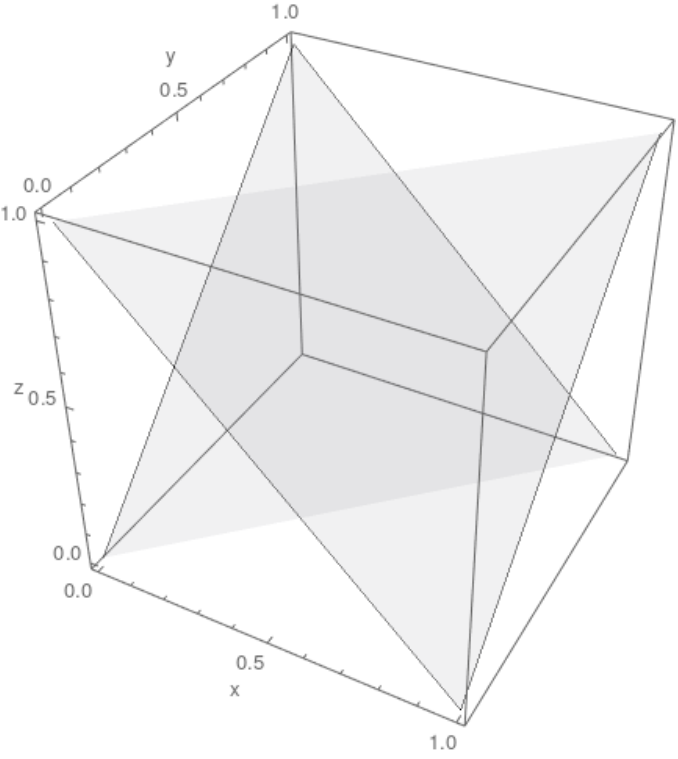}
    \subcaption{}
  \end{minipage}
  \begin{minipage}[b]{0.4\linewidth}
    \centering
    \includegraphics[scale=0.4]{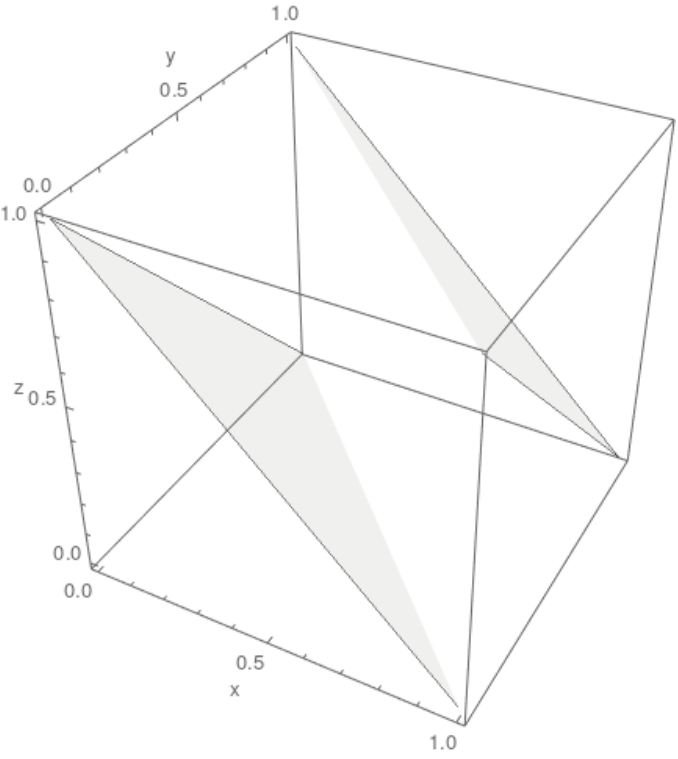}
    \subcaption{}
  \end{minipage}
\caption{Pictures of four planes: (a): $Z=(1+2^{23})X+Y \bmod{1}$ , 
(b): $X=(1+2^{23})X-Y \bmod{1}$, 
(c): $Z=-(1+2^{23})X+Y \bmod{1}$, 
(d): $Z=-(1+2^{23})X-Y \bmod{1}$ 
}
\label{fig:4planes}
\end{figure}

\begin{figure}[htbp]
  \centering
  \includegraphics[scale=0.4]{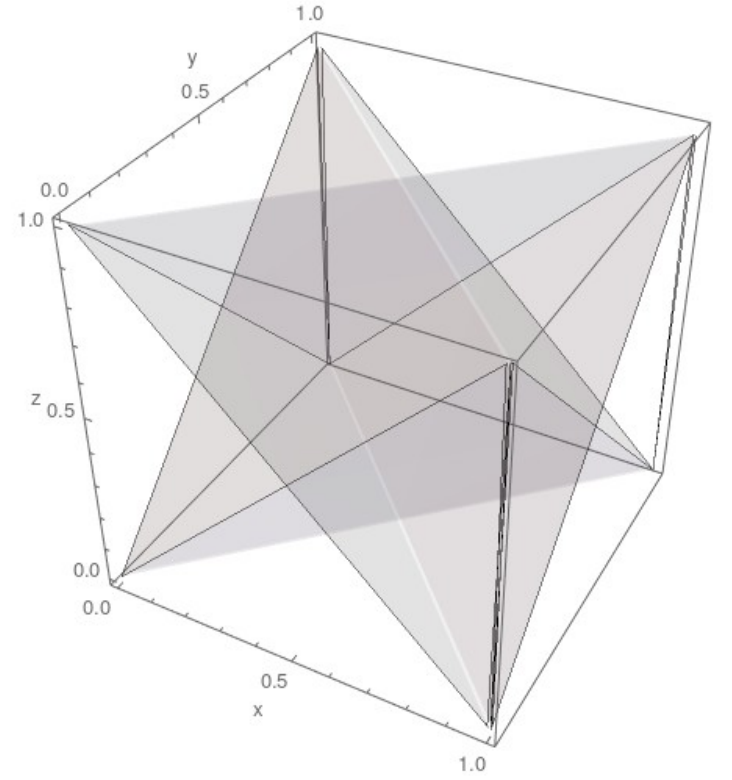}
  \caption{The union of four planes in Figure~\ref{fig:4planes}}
  \label{fig:union}
\end{figure}

Figure~\ref{fig:3Dplot2-23} shows the outputs of xorshift128+ with parameter $(a,b,c)=(23,17,26)$. 
Let $(X,Y,Z)$ be the consecutive outputs in $[0, 1)^3$. 
We only choose those with $X \leq 1/2^{23}$ and plot $(2^{23}X, Y, Z)$.
We repeat this until we obtain $10000$ points. 

\begin{figure}[htbp]
  \centering
  \includegraphics[scale=0.4]{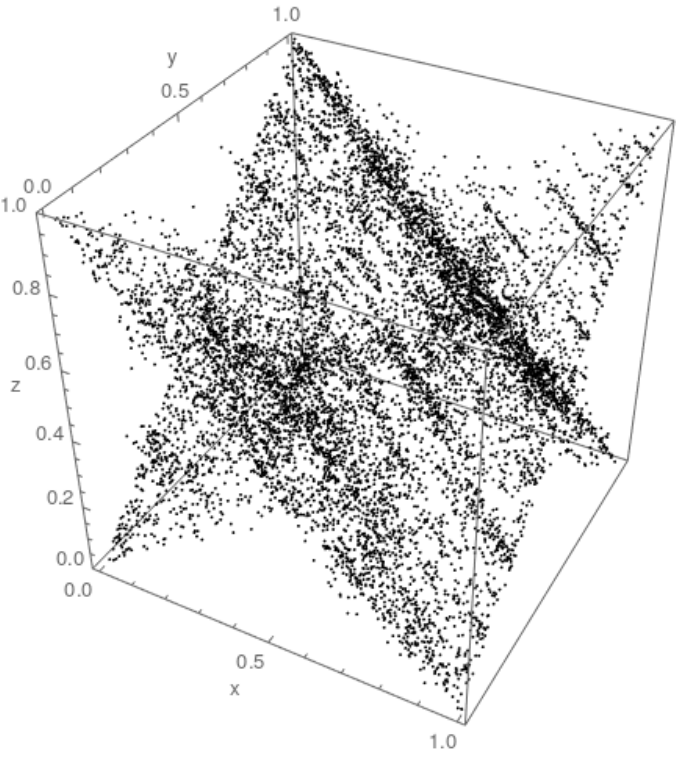}
  \caption{3D plots by xorshift128+: $X$-axis magnified by a factor of $2^{23}$}
  \label{fig:3Dplot2-23}
\end{figure}

Figure~\ref{fig:points-planes} shows both the four planes 
(Figure~\ref{fig:union}) and 
the outputs of xorshift128+ (Figure~\ref{fig:3Dplot2-23}). 
This coincidence supports the approximate analysis described in this section.

\begin{figure}[htbp]
  \centering
  \includegraphics[scale=0.4]{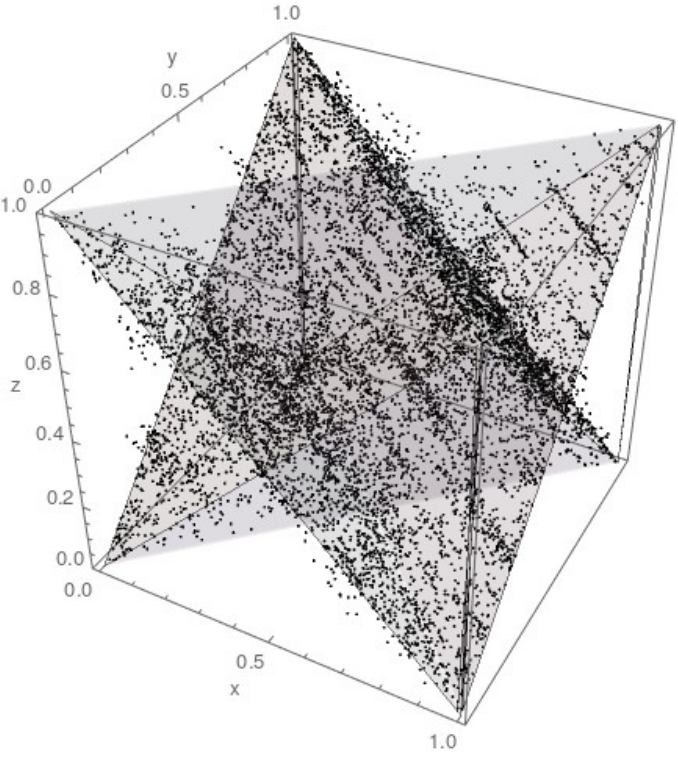}
  \caption{The union of Figure~\ref{fig:union} and Figure~\ref{fig:3Dplot2-23}}
  \label{fig:points-planes}
\end{figure}

\section*{Acknowledgement}
The authors thank the referees for their insightful comments.
This work was partly supported by JSPS KAKENHI Grant Numbers 
26310211, 15K13460, 17K14234, 18K03213, 19K03450 and by JST-CREST 151001.
Some parts of the work used computing resources from ISM, Research Numbers
2019-ISMCRP-05 and 2020-ISMCRP-0014.





\begin{thebibliography}{00}

\bibitem{MR2404400}
{\sc P.~L'Ecuyer and R.~Simard}, {\em {TestU01: a C library for empirical
  testing of random number generators}}, ACM Trans. Math. Software, 33 (2007),
pp.~Art. 22, 40.

\bibitem{LEMIRE}
{\sc D.~Lemire and M.E.~O'Neill}, {\em {Xorshift1024*, xorshift1024+, xorshift128+ and xoroshiro128+ fail statistical tests for linearity}},
Journal of Computational and Applied Mathematics, 350 (2019), 
pp.~139--142.
\url{https://doi.org/10.1016/j.cam.2018.10.019}

\bibitem{JSSv008i14}
{\sc G.~Marsaglia}, {\em Xorshift rngs}, Journal of Statistical Software,
 Articles, 8 (2003), pp.~1--6, 
\url{https://doi.org/10.18637/jss.v008.i14},
 \url{https://www.jstatsoft.org/v008/i14}.

\bibitem{xorshift-LEcuyer}
{\sc F.~Panneton and P.~L'Ecuyer}, {\em {On the xorshift random number generators}}, ACM Trans. Modeling and Computer Simulation, 15 (2005),pp.~346--361.
  ~Art. 22, 40.


\bibitem{VIGNA2017175}
{\sc S.~Vigna}, {\em Further scramblings of Marsaglia’s xorshift generators},
  Journal of Computational and Applied Mathematics, 315 (2017), pp.~175 -- 181,
  \url{http://www.sciencedirect.com/science/article/pii/S0377042716305301}.

\end{thebibliography}


\end{document}